\documentclass[%
 reprint,
groupedaddress,
 amsmath,amssymb,
 aps,
prx,
]{revtex4-2}

\usepackage{graphicx}
\usepackage{svg} 
\usepackage{dcolumn}
\usepackage{bm}
\usepackage{comment}
\usepackage{amsthm}
\usepackage{mathtools}
\usepackage{color}
\usepackage{algorithm2e}
\usepackage{amssymb}
\usepackage{booktabs}  
\usepackage{multirow}

\usepackage{braket}
\usepackage{subcaption}

\newtheorem{theorem}{\bf Theorem}[section]

\newtheorem{assumption}{\bf Assumption}[section]  
\newtheorem{Algorithm}{\bf Algorithm}[section]

\usepackage{bbm}
\usepackage{bm}

\newcommand\rms{\mathrm{s}}

\newcommand\rmu{\mathrm{u}}

\newcommand\rmU{\mathrm{U}}

\newcommand\bbr{\mathbb{R}}
\newcommand\bbs{\mathbb{S}}

\newcommand\bbu{\mathbb{U}}

\newcommand\calF{\mathcal{F}}

\begin{document}

\preprint{APS/123-QED}

\title{Model predictive quantum control:\\
A modular approach for efficient and robust quantum optimal control}

\author{Eya Guizani$^1$ and Julian Berberich$^2$}
 \affiliation{$^1$M.Sc. student, University of Stuttgart, 70569 Stuttgart, Germany}
 \affiliation{$^2$University of Stuttgart, Institute for Systems Theory and Automatic Control and Center for Integrated Quantum Science and Technology (IQST), 70569 Stuttgart, Germany}
 \email{julian.berberich@ist.uni-stuttgart.de}

\date{\today}

\begin{abstract}
Model predictive control (MPC) is one of the most successful modern control methods. It relies on repeatedly solving a finite-horizon optimal control problem and applying the beginning piece of the optimal input. In this paper, we develop a modular framework for improving efficiency and robustness of quantum optimal control (QOC) via MPC.
We first provide a tutorial introduction to basic concepts of MPC from a QOC perspective.
We then present multiple MPC schemes, ranging from simple approaches to more sophisticated schemes which admit stability guarantees.
This yields a modular framework which can be used 1) to improve efficiency of open-loop QOC and 2) to improve robustness of closed-loop quantum control by incorporating feedback.
We demonstrate these benefits with numerical results, where we benchmark the proposed methods against competing approaches.
\end{abstract}

\maketitle


\section{Introduction}

Quantum control is a key ingredient for the realization of quantum technologies~\cite{dong2010quantum,altafini2012modeling,glaser2015training,dong2022quantum,koch2022quantum}.
Quantum optimal control (QOC) is concerned with finding control inputs which make a quantum system behave optimally while respecting physical constraints, see~\cite{werschnik2007quantum,bonnard2012optimal,rembold2020introduction,boscain2021introduction,koch2022quantum,ansel2024introduction} for introductory and overview works.
QOC problems are computationally complex due to the exponential scaling of the system dimension.
Moreover, the bilinearity of the dynamics makes QOC a non-convex optimization problem for which an optimal solution can be hard to find. 
Motivated by these challenges, the literature contains various tailored optimal control methods for quantum systems, including gradient-based approaches such as gradient ascent pulse engineering (GRAPE)~\cite{khaneja2005optimal} and Krotov's method~\cite{reich2012monotonically}, or gradient-free methods such as chopped random basis (CRAB) optimization~\cite{caneva2011chopped,doria2011optimal,mueller2022one}.
Despite significant advances, the development of efficient and robust QOC methods remains an active field of research~\cite{koch2022quantum}.

In this paper, we introduce a modular framework for improving efficiency and robustness of QOC via model predictive control (MPC)~\cite{rawlings2020model}.
MPC is one of the most popular modern control methods with successful applications across countless domains including process control, robotics, automotive, aerospace, energy systems, and more.
The basic idea of MPC is to repeatedly solve optimal control problems over a fixed time horizon, to apply the beginning piece of the optimal input, and to repeat based on the newly measured state.
In this paper, we propose the model predictive quantum control (MPQC) framework which relies on repeatedly solving smaller QOC problems.
We present multiple MPQC schemes, ranging from a simple approach based on standard QOC problems to more sophisticated MPQC schemes which include additional constraints on the final state.
For the latter, we derive theoretical guarantees on convergence to the target state.
The proposed framework is modular in the sense that it can be combined with existing numerical optimization techniques for quantum systems including GRAPE, Krotov's method, and CRAB optimization.
We distinguish between two applications:
Open-loop and closed-loop (feedback) quantum control.
In open-loop control, the complete input trajectory is computed offline and then applied to the quantum system without intermediate measurements.
In closed-loop control, measurements are taken to obtain information about the quantum state, e.g., via state tomography~\cite{cramer2010efficient}, and thereby to adapt to uncertainties or noise.
MPQC can be applied to improve efficiency and robustness in both scenarios.
In open-loop control, the MPQC framework can be used to solve QOC problems more efficiently by partitioning them into smaller problems.
In closed-loop control, the MPQC framework allows us to integrate feedback into QOC and, in this way, achieve robustness against noise and model mismatch.
We demonstrate the efficiency and robustness benefits of MPQC with numerical examples.

The literature contains a number of existing works on using MPC for quantum control.
In~\cite{humaloja2018linear}, MPC is used to control systems governed by the infinite-dimensional Schr\"odinger equation.
The papers~\cite{lee2024robust,lee2024model,lee2025time} propose time-optimal MPC schemes for uncertain quantum systems, where measurements are used systematically to influence the system.
Further,~\cite{clouatre2022model} develops an MPC scheme based on a learned Hamiltonian, which is used to compute open-loop optimal control inputs.
On the other hand, the work~\cite{goldschmidt2022model} focuses on the specific problem of quantum state preparation by using MPC as a closed-loop controller.
While~\cite{humaloja2018linear,lee2024robust,lee2024model,lee2025time} address fundamentally different control problems, the latter two works~\cite{clouatre2022model} and~\cite{goldschmidt2022model} are related to our results in that they use MPC for open-loop and closed-loop QOC, respectively.
The present paper unifies and extends these works in several aspects:
More general systems dynamics including open quantum systems;
MPC schemes with terminal constraints;
rigorous theoretical analysis;
analysis of the influence of numerical optimization techniques;
unifying formulation of MPC which is applicable to both open-loop and closed-loop QOC.

The remainder of the paper is structured as follows.
In Section~\ref{sec:preliminaries}, we introduce the considered system dynamics and the quantum control problem.
Next, in Section~\ref{sec:MPQC}, we present the basic MPQC scheme.
Sections~\ref{sec:MPQC_TEC} and~\ref{sec:MPQC_setpoint} contain more sophisticated MPQC schemes with terminal equality constraints and setpoint optimization, respectively.
In Section~\ref{sec:MPQC_opt}, we discuss the influence of the choice of numerical optimization method and, in Section~\ref{sec:numerical_results}, we provide numerical results.
Finally, Section~\ref{sec:discussion} concludes the paper.

\section{Preliminaries}\label{sec:preliminaries}

In this section, we introduce preliminaries on the quantum system dynamics (Section~\ref{subsec:preliminaries_dynamics}) and the optimal control objective (Section~\ref{subsec:preliminaries_optimal}).

\subsection{Quantum system dynamics}\label{subsec:preliminaries_dynamics}
We study quantum control problems from a unified bilinear control systems viewpoint.
To be precise, we consider systems of the form 
\begin{equation}\label{eq:sys_X}
\dot{X}(t) = -\Big( A_0 + \sum_{j=1}^{m} u_j(t) B_j \Big) X(t),
\end{equation}
where \( X(t)\) is the generalized state variable, \( u_j(t) \) are real-valued control amplitudes, and $A_0$, $B_j$, $j=1,\dots,m$ are the system parameters.
Such bilinear control systems include a large class of relevant quantum systems~\cite{machnes2011comparing,berberich2024bringing}.
For example, pure state evolution under the time-dependent Schrödinger equation with a fixed drift Hamiltonian \( H_0 \) and time-dependent control Hamiltonians \( H_j, \,j= 1, \dots, m\) takes the form

\begin{equation}\label{eq:sys_psi}
\ket{\dot{\psi}(t)} = -i \Big( H_0 + \sum_{j=1}^m u_j(t) H_j \Big) |\psi(t)\rangle.
\end{equation}

This is equivalent to~\eqref{eq:sys_X} with $A_0=iH_0$, $B_j=iH_j$, $X(t)=\ket{\psi(t)}$.
Likewise, gate synthesis problems can be treated via

\begin{equation}\label{eq:sys_U}
\dot{U}(t) = -i \Big( H_0 + \sum_{j=1}^m u_j(t) H_j \Big) U(t),
\end{equation}

which reduces to~\eqref{eq:sys_X} with $A_0=iH_0$, $B_j=iH_j$, $X(t)=U(t)$.
Open quantum systems with Lindblad dynamics can be handled analogously by using $X(t)$ to represent a (vectorized) density matrix or quantum operation, see~\cite{machnes2011comparing,berberich2024bringing} for details.

For easier numerical implementation, we discretize time into uniform intervals of duration \( \Delta t \) and assume the control input is constant on each interval.
This results in piecewise-constant control functions $u(\tau)=u_t$ for $\tau\in[\Delta t\cdot t,\Delta t\cdot t+\Delta t)$ and $t=0,\dots,N-1$.
Here, $N$ is the discrete-time horizon satisfying $N\Delta t=T$.
Thus, we obtain the following evolution of $X$ in discrete time
\begin{align}\label{eq:sys_X_discrete}
    X_{t+1}=e^{-A(u_t)\Delta t}X_t
\end{align}
with
\begin{equation}
    A(u_t) = A_0 + \sum_{j=1}^m u_{t,j} B_j,
\end{equation}
where $u_{t,j}$ is the $j$-th input component at time $t$.

\subsection{Optimal control objective}\label{subsec:preliminaries_optimal}

In this paper, we consider the following quantum control objective:
Steering the state $X_t$ to a desired target $X_{\mathrm{ref}}$ by maximizing the fidelity $\calF(X_N,X_{\mathrm{ref}})$ at final time $N$.
The fidelity is given, e.g., by
\begin{align}\label{eq:fidelity_pure_state}
\mathcal{F}_\psi(\ket{\psi_N},\ket{\psi_{\mathrm{ref}}}) := |\langle \psi_\text{ref} | \psi_N \rangle|^2,
\end{align}
or
\begin{align}\label{eq:fidelity_unitary}
\mathcal{F}_\rmU(U_N,U_{\mathrm{ref}}) = \frac{1}{d^2} \left| \mathrm{tr}\left( U_\text{ref}^\dagger U_N \right) \right|^2,
\end{align}
in the pure state transfer or gate synthesis setup, respectively.
Here, we write $d$ for the dimension of the underlying Hilbert space.
In either case, the QOC problem amounts to finding a control input sequence $\{u_t\}_{t=0}^{N-1}$ that maximizes the corresponding fidelity for the final state while satisfying the system dynamics~\eqref{eq:sys_X_discrete}.
At the same time, we want to minimize the control effort in order to cope with practical constraints, e.g., on bandwidth, amplitude, and noise.
Hence, we minimize a weighted combination of the infidelity and the control cost, represented by the \emph{stage cost}
\begin{align}\label{eq:stage_cost}
    \ell(X_t,u_t)=\alpha(1-\calF(X_t,X_{\mathrm{ref}}))+\lVert u_t-\bar{u}_{\mathrm{ref},t}\rVert_R^2
\end{align}
for a scalar weight $\alpha\geq0$ and with the notation $\lVert u\rVert_R^2=u^\top R u$ for some positive semidefinite weighting matrix $R$.
Here, $\bar{u}_{\mathrm{ref},t}$ is a (possibly time-varying) input reference as frequently considered in QOC~\cite{reich2012monotonically}.
In practice, it can be chosen to be zero.
On the other hand, as we will see later in the paper, it can be beneficial for theoretical guarantees to have a constant input reference $\bar{u}_{\mathrm{ref},t}=u_{\mathrm{ref}}$, $t=0,\dots,N-1$, which is chosen such that $X_{\mathrm{ref}}$ is an eigenstate of the dynamics matrix (for closed quantum systems, the controlled Hamiltonian) with input $u_{\mathrm{ref}}$, i.e., 
\begin{align}\label{eq:eigenstate_X_ref_u_ref}
    A(u_{\mathrm{ref}})X_{\mathrm{ref}}=\lambda X_{\mathrm{ref}}
\end{align}
for some $\lambda\in\bbr$.
Further, we define the terminal cost function
\begin{align}
    \Phi(X_N)=\beta(1-\calF(X_N,X_{\mathrm{ref}})),
\end{align}
which penalizes the deviation of the final state from the target with the scalar weight $\beta\geq0$.
Moreover, we consider control constraints of the form $u_t\in\bbu$ for a given set $\bbu\subseteq\bbr^m$.
They can be used, e.g., to encode actuator limits $\lVert u_t\rVert\leq u_{\max}$, but also more general constraint sets or even no constraints at all, i.e., $\bbu=\bbr^m$.

We want to solve the following QOC problem, which aims to steer the system~\eqref{eq:sys_X_discrete} from the initial state $X_0$ to the target state $X_{\mathrm{ref}}$.
\begin{align}\label{eq:QOC}
    \min_{\{u_t\}_{t=0}^{N-1}}
    \>\>&\sum_{t=0}^{N-1}\ell(X_t,u_t)+\Phi(X_N)\\
    \mathrm{s.t.}\quad&X_{t+1}=e^{-A(u_t)\Delta t}X_t,\\
    &u_t\in\bbu,\>\>t=0,\dots,N-1.
\end{align}
The optimization problem~\eqref{eq:QOC} constitutes a standard QOC problem for which a variety of numerical optimization techniques and analytical insights exist, see~\cite{koch2022quantum} for an overview.
In this paper, we provide a modular strategy for tackling~\eqref{eq:QOC} based on MPC~\cite{rawlings2020model}.
To be precise, we construct solutions to~\eqref{eq:QOC} by repeatedly solving the problem over smaller time horizons and only using the beginning piece of the optimal input, see Section~\ref{sec:MPQC} for details.
We show that the proposed MPQC framework can guarantee exponential convergence to the target state $X_{\mathrm{ref}}$ and, thereby, implicitly provides a performant solution to the original QOC problem~\eqref{eq:QOC}.
Moreover, we discuss efficiency and robustness benefits of MPQC in comparison to standard QOC for open-loop and closed-loop quantum control problems.

\section{Model predictive quantum control}\label{sec:MPQC}

MPC is a receding-horizon control strategy in which a finite-horizon optimal control problem is solved repeatedly. Starting from the current state $X_t$, MPC predicts and optimizes over the system evolution for a future horizon of length $L$ based on the discrete-time dynamics~\eqref{eq:sys_X_discrete}.
Importantly, the prediction horizon $L$ is typically (significantly) shorter than the final time $N$.

\begin{subequations}\label{eq:MPQC}
In the following, we introduce the basic MPQC scheme.
At each time $t=0,\dots,N-1$, given the current state $X_t$, we solve the following quantum optimal control problem
\begin{align}\label{eq:MPQC_cost}
    \min_{\{\bar{u}_k(t)\}_{k=0}^{L-1}}
    \>\>&\sum_{k=0}^{L-1}\ell(\bar{X}_k(t),\bar{u}_k(t))+\Phi(\bar{X}_L(t))\\\label{eq:MPQC_dynamics}
    \mathrm{s.t.}\quad&\bar{X}_{k+1}(t)=e^{-A(\bar{u}_k(t))\Delta t}\bar{X}_k(t),\\
    \label{eq:MPQC_initial_state}
    &\bar{X}_0(t)=X_t,\\\label{eq:MPQC_constraints}
    &\bar{u}_k(t)\in\bbu,\>\>k=0,\dots,L-1.
\end{align}
\end{subequations}
Here, $\bar{u}_k(t)$ and $\bar{X}_k(t)$ denote the $k$-th step of the predicted input and state, appearing in the optimization problem at time $t$.
On the other hand, $X_t$ and $u_t$ refer to the state and input of the controlled system~\eqref{eq:sys_X_discrete} at time $t$ under the proposed MPC controller. 
Thus, the constraint~\eqref{eq:MPQC_initial_state} initializes the internal predicted state in the optimization problem with the current state $X_t$.
In problem~\eqref{eq:MPQC}, both $\bar{u}(t)$ and $\bar{X}(t)$ are optimization variables, but we only explicitly highlight the minimization over the control input $\bar{u}(t)$ since it uniquely determines the state $\bar{X}(t)$ and is therefore the only free variable.
We write 
\begin{align}
    \bar{X}^*(t)=\{\bar{X}_k^*\}_{k=0}^{L}\>\>\text{and}\>\>\bar{u}^*(t)=\{\bar{u}_k^*(t)\}_{k=0}^{L-1}
\end{align}
for the optimal solution of the optimization problem~\eqref{eq:MPQC} at time $t$.
In MPQC, once the optimal control sequence \(\bar{u}^*(t)\) is obtained, only the first $M$ steps of the optimal input are used for the QOC input sequence.
The process is repeated after $M$ time steps with updated state information, see Algorithm~\ref{alg:MPQC} for the detailed scheme.

\noindent\hrulefill

\begin{algorithm}[H]
\begin{Algorithm}\label{alg:MPQC}
\normalfont{\textbf{Model predictive quantum control}}\\
\textbf{Initialize} $t=0$ and \textbf{iterate:}
\begin{enumerate}
\item Based on the state $X_t$, solve~\eqref{eq:MPQC}.
\item Apply the first $M$ steps of the optimal input $u_{t+i}=\bar{u}_i^*(t)$, $i=0,\dots,M-1$, to the system~\eqref{eq:sys_X_discrete}.
\item Set $t=t+M$ and go back to Step 1.
\end{enumerate}
\end{Algorithm}
\end{algorithm}

\noindent\hrulefill

Algorithm~\ref{alg:MPQC} yields a control input sequence $\{u_t\}_{t=0}^{N-1}$ for the QOC problem~\eqref{eq:QOC} when stopping the iteration as soon as $t\geq N$.
It tackles the QOC problem~\eqref{eq:QOC} with horizon $N$ by repeatedly solving the smaller QOC problem~\eqref{eq:MPQC} with horizon $L$.
Note that, except for the different time horizons, the two problems are identical.
In this sense, the proposed framework is fully modular since we solve the original QOC problem~\eqref{eq:QOC} via a sequence of smaller QOC problems.
For the latter, existing numerical optimization techniques for QOC problems can be used, see~\cite{koch2022quantum}, but also more general nonlinear optimization methods.
In Section~\ref{sec:MPQC_opt}, we discuss the impact of the choice of optimization technique in more detail, in particular the theoretical guarantees when applying off-the-shelf QOC techniques such as GRAPE~\cite{khaneja2005optimal}, Krotov's method~\cite{reich2012monotonically}, or CRAB optimization~\cite{caneva2011chopped,doria2011optimal,mueller2022one}.

We distinguish between two application scenarios for the MPQC approach in Algorithm~\ref{alg:MPQC} as well as the other MPQC schemes introduced in Sections~\ref{sec:MPQC_opt} and~\ref{sec:MPQC_setpoint} below:
open-loop and closed-loop optimal control.
In open-loop MPQC, Algorithm~\ref{alg:MPQC} is used in an offline fashion to generate an input trajectory $\{u_t\}_{t=0}^{N-1}$ of length $N$ which is applied to the quantum systems without taking any measurements into account.
In this case, the state $X_t$ in Step 1.\ of Algorithm~\ref{alg:MPQC} is obtained by simulating the discrete-time dynamics~\eqref{eq:sys_X_discrete}.
Thereby, Algorithm~\ref{alg:MPQC} produces a (possibly suboptimal) candidate solution for the original QOC problem~\eqref{eq:QOC}.
In Section~\ref{sec:numerical_results}, we show with example systems that MPQC can be significantly more efficient than solving the problem~\eqref{eq:QOC} directly while keeping the performance at a comparable level.

In closed-loop MPQC, Algorithm~\ref{alg:MPQC} computes an input sequence $\{u_t\}_{t=0}^{N-1}$ by using state measurements of $X_t$ from the quantum system~\eqref{eq:sys_X_discrete} in Step 1.\ of Algorithm~\ref{alg:MPQC}.
A measurement of the state $X_t$ at time $t$ can be realized experimentally by repeatedly resetting the quantum system to the initial state $X_0$, applying the fixed parts of the control input $\{u_j\}_{j=0}^{t-1}$, and using measurements to construct the state $X_t$ via state tomography~\cite{cramer2010efficient}.
In this case, MPQC becomes a feedback-based quantum control method which combines the benefits of QOC with an increased robustness due to feedback.
In particular, superior performance can be obtained in the presence of model mismatch or noise.
Related approaches were proposed in the recent literature, e.g., GRAPE with feedback~\cite{porotti2023gradient} or quantum feedback control via reinforcement learning~\cite{guatto2024improving}, both of which can substantially enhance robustness but do not admit rigorous theoretical guarantees.

Finally, the value of $M$ is a user-chosen design parameter.
In classical MPC, $M=1$ is a common choice.
Larger values of $M$ reduce the computational effort because the optimization problem is solved less frequently, but also deteriorate the robustness of closed-loop implementations since fewer measurements are taken.

\section{Theoretical guarantees of MPQC via terminal equality constraints}\label{sec:MPQC_TEC}

The MPQC scheme presented in Section~\ref{sec:MPQC} relies on a basic MPC formulation~\cite{rawlings2020model}.
Although this MPC formulation often works well in practice, it does not admit theoretical guarantees in general and can even cause unstable (diverging) trajectories for the controlled system~\cite{raff2006nonlinear}.
There are two main approaches for enhancing MPC schemes such that they admit theoretical guarantees:
1) choosing a sufficiently long prediction horizon $L$, see~\cite{gruene2012nmpc} for explicit lower bounds,
and 2) adding terminal constraints on the final state $\bar{X}_L(t)$ to the optimization problem~\eqref{eq:sys_X_discrete}~\cite{rawlings2020model}.

In the remainder of the paper, we provide more sophisticated MPQC schemes which do admit theoretical guarantees and can, therefore, admit superior practical performance even with shorter prediction horizons $L$.
In particular, in the present section, we introduce an MPQC scheme with terminal equality constraints, which constitutes the simplest possibility for achieving theoretical guarantees in MPC~\cite{rawlings2020model}.
In Section~\ref{sec:MPQC_setpoint}, we present a scheme with setpoint optimization to improve the practical performance and reduce the computational complexity while keeping theoretical guarantees.

\begin{subequations}\label{eq:MPQC_TEC}
The MPQC scheme with terminal equality constraints is defined as follows:
At each time $t=0,\dots,N-1$, given the current state $X_t$, we solve the following optimal control problem
\begin{align}
\label{eq:MPQC_TEC_cost}
    \min_{\{\bar{u}_k(t)\}_{k=0}^{L-1}}
    \>\>&\sum_{k=0}^{L-1}\ell(\bar{X}_k(t),\bar{u}_k(t))\\
    \mathrm{s.t.}\quad&\bar{X}_{k+1}(t)=e^{-A(\bar{u}_k(t))\Delta t}\bar{X}_k(t),\\
    &\bar{u}_k(t)\in\bbu,\>\>k=0,\dots,L-1,\\\label{eq:MPQC_TEC_TEC}
    &\calF(\bar{X}_L(t),X_{\mathrm{ref}})=1.
\end{align}
\end{subequations}
In comparison to~\eqref{eq:MPQC}, the optimization problem~\eqref{eq:MPQC_TEC} contains the additional constraint~\eqref{eq:MPQC_TEC_TEC} which ensures that the fidelity between the final state $\bar{X}_L(t)$ and the target $X_{\mathrm{ref}}$ is one.
This implies that both are equal (modulo an unimportant global phase).
Note that we dropped the terminal cost function $\Phi(\bar{X}_L(t))$ from the cost~\eqref{eq:MPQC_TEC_cost} since it is zero due to~\eqref{eq:MPQC_TEC_TEC}.
We write $J(X_t,\bar{u}(t))$ for the cost of~\eqref{eq:MPQC_TEC} for a given input candidate $\bar{u}(t)=\{\bar{u}_k(t)\}_{k=0}^{N-1}$, and $J^*(X_t)$ for the optimal cost, i.e., $J^*(X_t)=J(X_t,\bar{u}^*(t))$ with the optimal control input $\bar{u}^*(t))$.
The optimization problem~\eqref{eq:MPQC_TEC} is used to generate an input for the quantum system~\eqref{eq:sys_X_discrete} analogous to Algorithm~\ref{alg:MPQC}, i.e., by solving the problem at time $t$ and only applying the beginning piece of the optimal input.
It is applicable in both open-loop and closed-loop implementations, compare the discussion in Section~\ref{sec:MPQC}.

We now provide a theoretical result for the performance under the MPQC scheme with terminal equality constraints.
In particular, we state an exponential stability property for the controlled system, i.e., the controlled state trajectory converges exponentially to the target $X_{\mathrm{ref}}$.
To this end, we introduce a norm $\lVert X\rVert$ for the state, which we assume to be bounded by the fidelity $\calF$.
This assumption, together with additional technical assumptions, is detailed in the following.
\begin{assumption}\label{ass:TEC}
    \begin{enumerate}
        \item 
It holds that
\begin{align}\label{eq:fidelity_norm_equivalence}
    \lVert X_1-X_2\rVert^2\leq c_1(1-\calF(X_1,X_2))
\end{align}
for some $c_1>0$ and any $X_1$, $X_2$.

        \item The input constraint set $\bbu$ is compact.
        
        \item The stage cost~\eqref{eq:stage_cost} is positive definite in the state, i.e., it holds that $\alpha>0$.
        
        \item The input reference is constant $\bar{u}_{\mathrm{ref},t}=u_{\mathrm{ref}}$, $t=0,\dots,N-1$, and it satisfies~\eqref{eq:eigenstate_X_ref_u_ref}.

        \item The optimal cost of the MPQC problem~\eqref{eq:MPQC_TEC} satisfies
        \begin{align}\label{eq:MPQC_TEC_cost_upper_bound}
            J^*(X)\leq c_{\rmu}\lVert X-X_{\mathrm{ref}}\rVert^2
        \end{align}
        with some $c_{\rmu}>0$ and for all $X$ for which~\eqref{eq:MPQC_TEC} is feasible.
    \end{enumerate}
\end{assumption}

The first assumption involving the inequality~\eqref{eq:fidelity_norm_equivalence} is satisfied for most relevant QOC scenarios, e.g., for pure states with the diamond norm and the fidelity~\eqref{eq:fidelity_pure_state}, for unitaries with the Frobenius norm and the fidelity~\eqref{eq:fidelity_unitary}, but also for mixed states and quantum operations with suitably defined norms and fidelities, compare~\cite{gilchrist2005distance}.

The second assumption on compactness of $\bbu$ is not restrictive since experimental inputs are typically bounded.
The third assumption on $\alpha>0$ can be easily ensured since $\alpha$ is a user-chosen parameter, compare~\eqref{eq:stage_cost}.
The fourth assumption involving~\eqref{eq:eigenstate_X_ref_u_ref} reduces our analysis to target states $X_{\mathrm{ref}}$ which are eigenstates of the controlled Hamiltonian for some control input $u_{\mathrm{ref}}$.
The assumption is necessary for obtaining rigorous guarantees via standard MPC arguments, which are asymptotic by nature and therefore require that the system can be kept at the target state.
Note that the target input $u_{\mathrm{ref}}$ can be computed for a given target state $X_{\mathrm{ref}}$ based on~\eqref{eq:eigenstate_X_ref_u_ref} since this equation is linear in the input.
Finally, the fifth assumption is also common in the MPC literature and is connected to controllability since it requires the ability to steer the system to $X_{\mathrm{ref}}$ with suitably bounded cost~\cite{rawlings2020model}.
We impose a quadratic upper bound in~\eqref{eq:MPQC_TEC_cost_upper_bound} to derive exponential stability in Theorem~\ref{thm:MPQC_TEC}, but we note that it can be relaxed to a more general continuous function, in which case the analysis below implies asymptotic stability at a potentially non-exponential rate.
For a detailed treatment of controllability properties of quantum systems, we refer to~\cite{dalessandro2021quantum}.

\begin{theorem}\label{thm:MPQC_TEC}
    Suppose Assumption~\ref{ass:TEC} holds, $M=1$, and the MPQC problem~\eqref{eq:MPQC_TEC} is feasible at time $t=0$.
Consider the system~\eqref{eq:sys_X_discrete} controlled via Algorithm~\ref{alg:MPQC}.

The target state $X_{\mathrm{ref}}$ is exponentially stable for the controlled system, i.e., there exist constants $C>0$, $0<\gamma<1$ such that, for any $t\geq0$,
\begin{align}\label{eq:thm_MPQC_TEC}
    1-\calF(X_t,X_{\mathrm{ref}})\leq C\gamma^t(1-\calF(X_0,X_{\mathrm{ref}})).
\end{align}
\end{theorem}
\begin{proof}
Using classical MPC arguments~\cite{rawlings2020model}, one can show that the optimization problem~\eqref{eq:MPQC_TEC} is feasible at any time $t\geq0$ when it is feasible at initial time $t=0$.
In particular, one obtains the bound
\begin{align}\label{eq:thm_MPQC_TEC_proof1}
    J^*(X_{t+1})-J^*(X_t)\leq-\ell(X_t,u_t)
\end{align}
for any $t=0,\dots,N-1$, see~\cite{rawlings2020model} for details.
The definition of the stage cost in~\eqref{eq:stage_cost} implies
\begin{align}\label{eq:thm_MPQC_TEC_proof1b}
    J^*(X_{t+1})-J^*(X_t)\leq-\alpha(1-\calF(X_t,X_{\mathrm{ref}})).
\end{align}
Combining~\eqref{eq:fidelity_norm_equivalence} and~\eqref{eq:MPQC_TEC_cost_upper_bound}, we obtain 
\begin{align}\label{eq:thm_MPQC_TEC_proof2}
    J^*(X_t)\leq c_\rmu c_1(1-\calF(X_t,X_{\mathrm{ref}}).
\end{align}
Plugging this into~\eqref{eq:thm_MPQC_TEC_proof1b}, we infer 
\begin{align}
    J^*(X_{t+1})\leq\underbrace{(1-\frac{\alpha}{c_{\rmu}c_1})}_{\gamma=}J^*(X_t).
\end{align}
Hence, we have 
\begin{align}\label{eq:thm_MPQC_TEC_proof3}
    J^*(X_t)\leq\gamma^t J^*(X_0).
\end{align}
Note that $0<\gamma<1$.
By definition of the cost~\eqref{eq:MPQC_TEC_cost}, we have $\alpha(1-\calF(X_t,X_{\mathrm{ref}})\leq J^*(X_t)$.
This implies 
\begin{align}
    \alpha(1-\calF(X_t,X_{\mathrm{ref}})\leq &J^*(X_t)\stackrel{\eqref{eq:thm_MPQC_TEC_proof3}}{\leq} \gamma^tJ^*(X_0)\\\nonumber 
    \stackrel{\eqref{eq:thm_MPQC_TEC_proof2}}{\leq}
    &c_\rmu c_1\gamma^t(1-\calF(X_0,X_{\mathrm{ref}})
\end{align}
for any $t=0,\dots,N-1$, which proves~\eqref{eq:thm_MPQC_TEC} with $C=\frac{c_\rmu c_1}{\alpha}$.
\end{proof}

Theorem~\ref{thm:MPQC_TEC} states that the infidelity between $X_t$ and $X_{\mathrm{ref}}$ decays exponentially and, therefore, $X_t$ converges exponentially to $X_{\mathrm{ref}}$ under the MPQC scheme with terminal equality constraints.
Thus, MPQC with terminal equality constraints yields a solution of the QOC problem~\eqref{eq:QOC} based on the repeated solution of smaller QOC problems with theoretical guarantees.
The proof of Theorem~\ref{thm:MPQC_TEC} follows elementary MPC arguments with the main challenge of using the fidelity-based stage cost $\ell$, which is not standard in the classical MPC literature.
The result assumes $M=1$ for simplicity, i.e., the optimization problem~\eqref{eq:MPQC_TEC} is solved at each time step $t=0,\dots,N-1$, but analogous results can be derived for $M>1$.
Beyond guaranteeing exponential stability under the MPQC scheme~\eqref{eq:MPQC_TEC}, the approach introduced in the present section provides the basis for the more advanced MPQC formulation presented in Section~\ref{sec:MPQC_setpoint}.

\section{MPQC with setpoint optimization}\label{sec:MPQC_setpoint}

While MPQC with terminal equality constraints as in Section~\ref{sec:MPQC_TEC} guarantees exponential stability, it admits several drawbacks.
Due to the terminal constraint~\eqref{eq:MPQC_TEC_TEC}, the optimization problem is only feasible at $t=0$ when the system can be steered from the initial state $X_0$ to the target state $X_{\mathrm{ref}}$ within $L$ steps.
This is only possible for initial states $X_0$ close to $X_{\mathrm{ref}}$ or for sufficiently large prediction horizons $L$.
Moreover, an input $u_{\mathrm{ref}}$ needs to be available for which $X_{\mathrm{ref}}$ is an eigenstate of the corresponding matrix $A(u_{\mathrm{ref}})$, compare~\eqref{eq:eigenstate_X_ref_u_ref}.
Further, MPC schemes with terminal equality constraints can admit poor robustness and performance.

In the following, we present an alternative, more advanced MPQC scheme which overcomes these drawbacks via setpoint optimization.
The approach relies on the MPC for tracking framework~\cite{limon2008mpc,limon2018nonlinear,koehler2020nonlinear}.
The key idea is to relax the terminal equality constraint~\eqref{eq:MPQC_TEC} by making the target setpoint an optimization variable, which can be an arbitrary steady-state for the system dynamics~\eqref{eq:sys_X_discrete}.
The difference between this \emph{artificial} setpoint and the actual \emph{target} setpoint $(X_{\mathrm{ref}},u_{\mathrm{ref}})$ is then penalized in the cost to ensure that the controlled system converges to the target setpoint.

\begin{subequations}\label{eq:MPQC_tracking}
We now introduce the MPQC scheme with setpoint optimization.
At each time $t=0,\dots,N-1$, given the current state $X_t$, we solve the following optimal control problem
\begin{align}\nonumber
    \min_{
    \substack{\{\bar{u}_k(t)\}_{k=0}^{L-1}\\X^\rms(t),u^\rms(t)}
    }
    \>\>&\sum_{k=0}^{L-1}
    \alpha(1-\calF(X_t,X^\rms(t)))
    +\lVert u_t-u^\rms(t)\rVert_R^2\\
\label{eq:MPQC_tracking_cost}
    &+\eta(1-\calF(X^\rms(t),X_{\mathrm{ref}}))
    +\lVert u^\rms(t)-u_{\mathrm{ref}}\rVert_S^2\\
    \mathrm{s.t.}\quad&\bar{X}_{k+1}(t)=e^{-A(\bar{u}_k(t))\Delta t}\bar{X}_k(t),\\
    &\bar{u}_k(t)\in\bbu,\>\>k=0,\dots,L-1,\\\label{eq:MPQC_tracking_TEC}
    &\calF(\bar{X}_L(t),X^\rms(t))=1,\\\label{eq:MPQC_tracking_steady_state}
    &(X^\rms(t),u^\rms(t))\in\bbs.
\end{align}
\end{subequations}

The main difference to the optimization problem~\eqref{eq:MPQC_TEC} from Section~\ref{sec:MPQC_TEC} is the introduction of the artificial setpoint $(X^\rms(t),u^\rms(t))$, which is an optimization variable and therefore depends on the time step $t$ at which the optimization problem~\eqref{eq:MPQC_tracking} is solved.
The terminal equality constraint~\eqref{eq:MPQC_tracking_TEC} is now taken w.r.t.\ this artificial setpoint.
On the other hand, the cost~\eqref{eq:MPQC_tracking_cost} penalizes the distance from $(X^\rms(t),u^\rms(t))$ to the actual target setpoint $(X_{\mathrm{ref}},u_{\mathrm{ref}})$ with weighting parameter $\eta>0$ and a positive semidefinite matrix $S$.

The set $\bbs$ in~\eqref{eq:MPQC_tracking_steady_state} is the steady-state manifold of the system~\eqref{eq:sys_X_discrete}, which is defined as 
\begin{align}\label{eq:steady_state_manifold}
    \bbs=\{(X^\rms,u^\rms)|X^\rms=e^{-A(u^\rms)}X^\rms\}.
\end{align}
Hence, the constraint~\eqref{eq:MPQC_tracking_steady_state} implies that $X^\rms(t)$ is a steady-state for the controlled system with control input $u^\rms(t)$.
The optimization problem~\eqref{eq:MPQC_tracking} is used for constructing a QOC input sequence $\{u_t\}_{t=0}^{N-1}$ solving problem~\eqref{eq:QOC} in an MPC fashion precisely as in Sections~\ref{sec:MPQC} and~\ref{sec:MPQC_TEC}.
To be precise, we solve the optimization problem with initial state $X_t$, store the first $M$ steps of the optimal control input, and repeat, compare Algorithm~\ref{alg:MPQC}.

MPQC with setpoint optimization (i.e., the MPQC scheme based on problem~\eqref{eq:MPQC_tracking}) has several important advantages over the approach from Section~\ref{sec:MPQC_TEC}.
Due to the optimization over the artificial setpoint $(X^\rms(t),u^\rms(t))$ in~\eqref{eq:MPQC_tracking}, one can typically use significantly smaller values for the prediction horizon $L$.
This can lead to a substantial computational speedup.
Moreover, note that we allow $S=0$, in which case the MPQC problem is independent of $u_{\mathrm{ref}}$.
This brings the practical advantage that one only has to provide a target state $X_{\mathrm{ref}}$ without having to compute a corresponding target input.

The MPC for tracking approach, on which the MPQC scheme in the present section relies, admits a solid theoretical foundation.
Under suitable assumptions on the system dynamics and the cost and constraint parameters, one can prove desirable theoretical properties for the controlled system such as exponential stability, see~\cite{krupa2024model} for a recent introduction.
Most of the arguments carry over directly from the classical to the quantum setting.
These results typically assume convexity of the set $\bbs$ in~\eqref{eq:steady_state_manifold}, which holds for mixed states.
On the other hand, it is violated for pure states or unitary operators, in which case appropriate modifications need to be taken to account for the non-convex steady-state manifold, see~\cite{soloperto2022nonlinear}.

\section{Numerical optimization techniques}\label{sec:MPQC_opt}

Applying the MPQC schemes from Sections~\ref{sec:MPQC}, \ref{sec:MPQC_TEC}, and~\ref{sec:MPQC_setpoint} requires the repeated solution of QOC problems with a shorter horizon.
In the following, we discuss different numerical optimization techniques which can be used to this end, as well as their role in our theoretical analysis.

In general, the MPQC optimization problems~\eqref{eq:MPQC}, \eqref{eq:MPQC_TEC}, and~\eqref{eq:MPQC_tracking} are nonlinear optimization problems, for which standard methods can be used, e.g., the CasADi framework~\cite{andersson2019casadi} along with solvers based on sequential quadratic programming~\cite{boggs1995sequential} or interior point methods~\cite{waechter2005implementation}.
There also exist tailored numerical optimization techniques for MPC~\cite{diehl2005nominal,diehl2009efficient,liao2020time,zanelli2021lyapunov}, which equally apply in the present MPQC setup.

On the other hand, one can also resort to dedicated numerical optimization techniques for QOC.
In particular, the computational challenges of QOC have led to the development of tailored solution techniques such as GRAPE~\cite{khaneja2005optimal}, Krotov's method~\cite{reich2012monotonically}, and CRAB optimization~\cite{caneva2011chopped,doria2011optimal,mueller2022one}.
The key benefit of the proposed framework is that it is completely modular in this respect.
Our approach relies on breaking down the (possibly large) QOC problem~\eqref{eq:QOC} into a sequence of QOC problems with smaller time horizon ($L$ instead of $N$), which can be tackled using common techniques from classical optimal control or QOC.
This means that our framework does not only exploit the rich existing literature on numerical methods for nonlinear optimization, classical optimal control, and QOC, but also that future advances on solving QOC problems more efficiently will be applicable to improve efficiency in our framework as well.
This includes, for example, the recent QOC algorithms based on geodesic pulse engineering~\cite{lewis2025quantum}, polynomial optimization~\cite{gaggioli2025unitary}, and low-rank models~\cite{goutte2025low}.

However, the theoretical guarantees under MPQC may depend on the choice of solution method.
Classical numerical optimization techniques are very flexible in terms of the stage cost, the terminal cost, constraints, and additional decision variables such as artificial setpoints.
As a result, it is straightforward to implement the MPQC schemes from Sections~\ref{sec:MPQC_TEC} and~\ref{sec:MPQC_setpoint} such that, e.g., Theorem~\ref{thm:MPQC_TEC} can be used to guarantee stability and convergence properties.
On the other hand, common QOC approaches are partially less flexible and typically cannot be directly used to implement the optimization problems from Sections~\ref{sec:MPQC_TEC} and~\ref{sec:MPQC_setpoint}.
In the following, we comment on the theoretical guarantees that are provided when using common QOC methods to solve the optimization problems arising in MPQC.

We begin by discussing GRAPE~\cite{khaneja2005optimal}, which is a numerical optimization technique for QOC problems of the form~\eqref{eq:QOC}.
The most basic form of GRAPE~\cite{khaneja2005optimal} does not consider a stage cost or input constraints, i.e., it solves the problem~\eqref{eq:QOC} with $\ell(X_t,u_t)=0$ and $\bbu=\bbr^m$.
With simple modifications, GRAPE can solve QOC problems including a stage cost~\cite{fauquenot2025open} as well as non-trivial input constraints (e.g., via projection).
However, GRAPE does not realize a terminal equality constraint as in~\eqref{eq:MPQC_TEC_TEC}, which is required to derive theoretical guarantees in Section~\ref{sec:MPQC_TEC}.
Nevertheless, the existing MPC literature contains a variety of more sophisticated stability results which are still applicable in this scenario.
In the following, we discuss three possible approaches.

For the first approach, we assume that any state $X$ can be made a steady-state, i.e., for any $X^\rms$ there exists $u^\rms$ such that 
\begin{align}
    X^\rms=e^{-A(u^\rms)}X^\rms.
\end{align}
Then, the optimization problem~\eqref{eq:MPQC_tracking} can be simplified to the form~\eqref{eq:MPQC}, which is amenable to GRAPE.
In particular, setting $S=0$, one can drop the constraint~\eqref{eq:MPQC_tracking_steady_state}.
Moreover, the term 
\begin{align}
    \eta(1-\calF(X^\rms(t),X_{\mathrm{ref}}))=\eta(1-\calF(\bar{X}_L(t),X_{\mathrm{ref}}))
\end{align}
then becomes a standard terminal cost as in~\eqref{eq:QOC}.
Thus, the resulting optimization problem can indeed be solved via GRAPE, and it relies on a rigorous theoretical foundation, compare the discussion in Section~\ref{sec:MPQC_setpoint}.
Second, the work~\cite{limon2006stability} shows that stability of MPC can be guaranteed with a terminal cost and without terminal constraints such as~\eqref{eq:MPQC_TEC_TEC}, i.e., in a setup amenable to GRAPE, under suitable assumptions on the terminal cost.
In particular, the terminal cost is required to be a Lyapunov function for the controlled system.
Deriving terminal cost functions for MPQC, e.g., based on quantum Lyapunov control~\cite{cong2013survey}, is an interesting issue for future research.
Third, MPC without terminal constraints guarantees stability of the controlled system assuming that the prediction horizon $L$ is sufficiently long and that suitable controllability properties hold~\cite{gruene2012nmpc}.

Next, we discuss Krotov's method~\cite{reich2012monotonically}, which is directly applicable to the basic QOC problem~\eqref{eq:QOC} and, hence, the MPQC problem~\eqref{eq:MPQC} from Section~\ref{sec:MPQC}.
It should be emphasized that common formulations of Krotov's method include time-dependent cost weights, which can be tuned to encourage physically reasonable control shapes.
On the other hand, the MPQC formulations in Sections~\ref{sec:MPQC}--\ref{sec:MPQC_setpoint} as well as the theoretical result in Theorem~\ref{thm:MPQC_TEC} only allow for time-independent weights $\alpha$ and $R$, compare~\eqref{eq:stage_cost}.
Further, Krotov's method does not enforce a terminal equality constraint such that the theoretical results from Section~\ref{sec:MPQC_TEC} do not directly apply.
Instead, more sophisticated MPC arguments need to be employed to prove theoretical guarantees, compare the discussion for GRAPE above.

Finally, we address CRAB optimization~\cite{caneva2011chopped,doria2011optimal,mueller2022one}, where the control input is parameterized via basis functions whose coefficients can be optimized using existing solvers.
Analogous to GRAPE, the basic CRAB formulation includes neither input constraints nor a stage cost, but an extension to addressing both is straightforward.
Thus, CRAB can be used to solve the basic QOC problem~\eqref{eq:QOC}.
As for GRAPE and Krotov's method, terminal equality constraints (Section~\ref{sec:MPQC_TEC}) cannot be directly handled using CRAB, which requires more sophisticated tools for a rigorous theoretical analysis, see above.

In summary, GRAPE, Krotov's method, and CRAB optimization can all be used to solve the basic QOC problem~\eqref{eq:QOC} and, hence, to implement the MPQC scheme explained in Section~\ref{sec:MPQC}.
However, contrary to generic nonlinear optimization solvers, none of them can handle terminal equality constraints (Section~\ref{sec:MPQC_TEC}) or artificial setpoints (Section~\ref{sec:MPQC_setpoint}).
Hence, deriving rigorous theoretical guarantees in this case requires more sophisticated arguments from MPC theory and is an interesting issue for future research.

\section{Numerical results}\label{sec:numerical_results}

In this section, we apply the developed MPQC framework, compare it to alternative QOC approaches, and study the influence of different optimization techniques.
In Section~\ref{subsec:investigation_of_the_basic_mpqc_scheme}, we apply the basic MPQC scheme explained in Section~\ref{sec:MPQC} to a single-qubit state transfer problem.
Next, we present results for the advanced MPQC formulations with terminal equality constraints (Section~\ref{subsec:importance_of_terminal_equality_constraint}) and setpoint optimization (Section~\ref{subsec:tracking_mpc}).
Further, in Section~\ref{subsec:numerical_implementation_comparison}, we compare the proposed MPQC approach to an existing QOC technique.
Finally, in Section~\ref{subsec:robustness_of_mpqc}, we demonstrate the improved robustness when using MPQC for closed-loop quantum control.
In Appendix~\ref{app:numerical_implementation}, we discuss details on the numerical implementation, in particular on nonlinear optimization with CasADi~\cite{andersson2019casadi}.

\subsection{Basic MPQC Scheme}
\label{subsec:investigation_of_the_basic_mpqc_scheme}

In the following, we apply the basic MPQC scheme presented in Section~\ref{sec:MPQC} when using different numerical optimization techniques for solving the optimization problem~\eqref{eq:MPQC}.
We consider the task of preparing a desired pure quantum state $\ket{\psi_{\mathrm{ref}}}$ from a known initial state for a two-level quantum system with Hamiltonian
\begin{align}
    H(t) = H_0 + u_1(t)\,H_1 + u_2(t)\,H_2 + u_3(t)\,H_3.
\end{align}
Here, $H_0 = \omega \sigma_z = -0,5 \sigma_z$ represents the drift Hamiltonian, and 
$H_1 = \sigma_x$, $H_2 = \sigma_y$, $H_3 = \sigma_z$ are the control Hamiltonians. 
The control amplitudes are subject to box constraints $u_j(t)\in[-1,1]$ for $j=1,2,3$. 
We set the initial state of the qubit to 
    $\ket{\psi_0} = \ket{0}$.
    The above Hamiltonian takes the form
\begin{align}
    H(t) = u_1(t)\sigma_x + u_2(t)\sigma_y + (\omega + u_3(t))\sigma_z,
\end{align}
so that any pure state can be realized as an eigenstate of the Hamiltonian through appropriate constant controls, i.e., \eqref{eq:eigenstate_X_ref_u_ref} can be satisfied.

The control performance is measured in terms of the final fidelity $\big| \langle \psi_{\mathrm{ref}}, \psi(T) \rangle \big|^2$.
The computational efficiency is quantified by the total runtime required to compute the QOC control sequence over the entire time horizon, which includes the repeated solution of the MPQC problem~\eqref{eq:MPQC}. 
For numerical implementation, the system is discretized over a total evolution time of $T = 5$~ns with $N = 100$ timesteps, resulting in a timestep $\Delta t = T/N = 0.05$ns.

The prediction horizon is set to $L=10$ time steps. 
Further, throughout all numerical results in Section~\ref{sec:numerical_results}, the stage cost parameters in~\eqref{eq:stage_cost} are chosen as $\alpha = 1$ and $R = 10^{-4} I$.
We solve the problem with CasADi and solver IPOPT (Appendix~\ref{app:numerical_implementation}) and with GRAPE.
The standard GRAPE algorithm computes the gradient of only the terminal cost with respect to the control inputs.
Thus, we employ a modified version which computes the gradient of the full cost~\eqref{eq:MPQC_cost}, including both the stage cost and terminal cost.

The 
step size in each gradient step 
is set to $0.2$. 
Moreover, the input constraints are ensured via projection.

\begin{table}[h]
\centering
\caption{Basic MPQC: Comparison of CasADi with solver IPOPT and GRAPE for different target states}
\label{tab:qoc_comparison1}
\begin{tabular}{lccc}
\toprule
\textbf{Target} & \textbf{Method} & \textbf{Final fidelity} & \textbf{Runtime [s]} \\
\midrule
\multirow{3}{*}{$\ket{1}$} & GRAPE & 0.999239 & 92.78 \\
 & CasADi IPOPT & 1.000000 & 4.74 \\
\midrule
\multirow{3}{*}{$\ket{+}$} & GRAPE & 0.994909 & 71.46 \\
 & CasADi IPOPT & 0.999733 & 5.31 \\
\midrule
\multirow{3}{*}{$\ket{-}$} & GRAPE & 0.995364 & 70.05 \\
 & CasADi IPOPT & 0.999733 & 5.44 \\
\bottomrule
\end{tabular}
\end{table}

Table~\ref{tab:qoc_comparison1} summarizes the final fidelity and average runtime when using the basic MPQC scheme with CasADi and GRAPE for different target states.
With either optimization method, MPQC yields a high final fidelity.
For this example, CasADi achieves superior performance in comparison to GRAPE while admitting a smaller runtime.

\subsection{MPQC with terminal equality constraint}
\label{subsec:importance_of_terminal_equality_constraint}

In the following, we apply the terminal equality constrained (TEC) MPQC scheme from Section~\ref{sec:MPQC_TEC} and compare it to the basic MPQC scheme from Section~\ref{sec:MPQC}.
Since GRAPE and Krotov's method do not allow one to implement a terminal equality constraint (compare Section~\ref{sec:MPQC_opt}), we only employ CasADi with IPOPT.
We consider a state transfer problem for the single-qubit system governed by the Hamiltonian
\begin{equation}\label{eq:numerical_results_TEC_Hamiltonian}
H(t) = \omega \, \sigma_z + u(t) \, \sigma_x.
\end{equation}

The control objective involves transferring the quantum state from 
\begin{equation}
\ket{+} = \frac{1}{\sqrt{2}} (\ket{0} + \ket{1})
\end{equation}
to 
\begin{equation}
\ket{-} = \frac{1}{\sqrt{2}} (\ket{0} - \ket{1}).
\end{equation} 
Note that this example 
violates the condition~\eqref{eq:eigenstate_X_ref_u_ref}, i.e., 
the target state is \emph{not} an eigenstate of the Hamiltonian $H(u_{\mathrm{ref}})$ for any constant reference input $u_{\mathrm{ref}}$. 
To ensure feasibility of the TEC MPQC formulation in this setting, the prediction horizon must be chosen sufficiently large. In particular, we set $L = 30$.

Both the basic MPQC scheme from Section \ref{sec:MPQC} and the TEC MPQC approach from \ref{sec:MPQC_TEC} are applied to the single-qubit state transfer task, and their performance is compared in terms of final fidelity and corresponding runtimes, compare Table~\ref{tab:TEC_comparison}. 
\begin{table}[h!]
\centering
\caption{Basic vs. TEC MPQC for single-qubit state transfer}
\label{tab:TEC_comparison}
\begin{tabular}{lccc}
\hline
\textbf{MPC Scheme} & \textbf{Final Fidelity} & \textbf{Runtime[s]} & \\
\hline
basic & 0.865947 & 8.40\\
TEC & 1.000000 & 62.06 \\
\hline
\end{tabular}
\end{table}

The TEC MPQC scheme achieves exact convergence to the target state. In contrast, the basic MPQC approach 
does not achieve exact state transfer. Notably, increasing the control or prediction horizon in the unconstrained scheme does not lead to exact convergence.
\begin{figure}[t]              \includegraphics[width=0.52\textwidth]{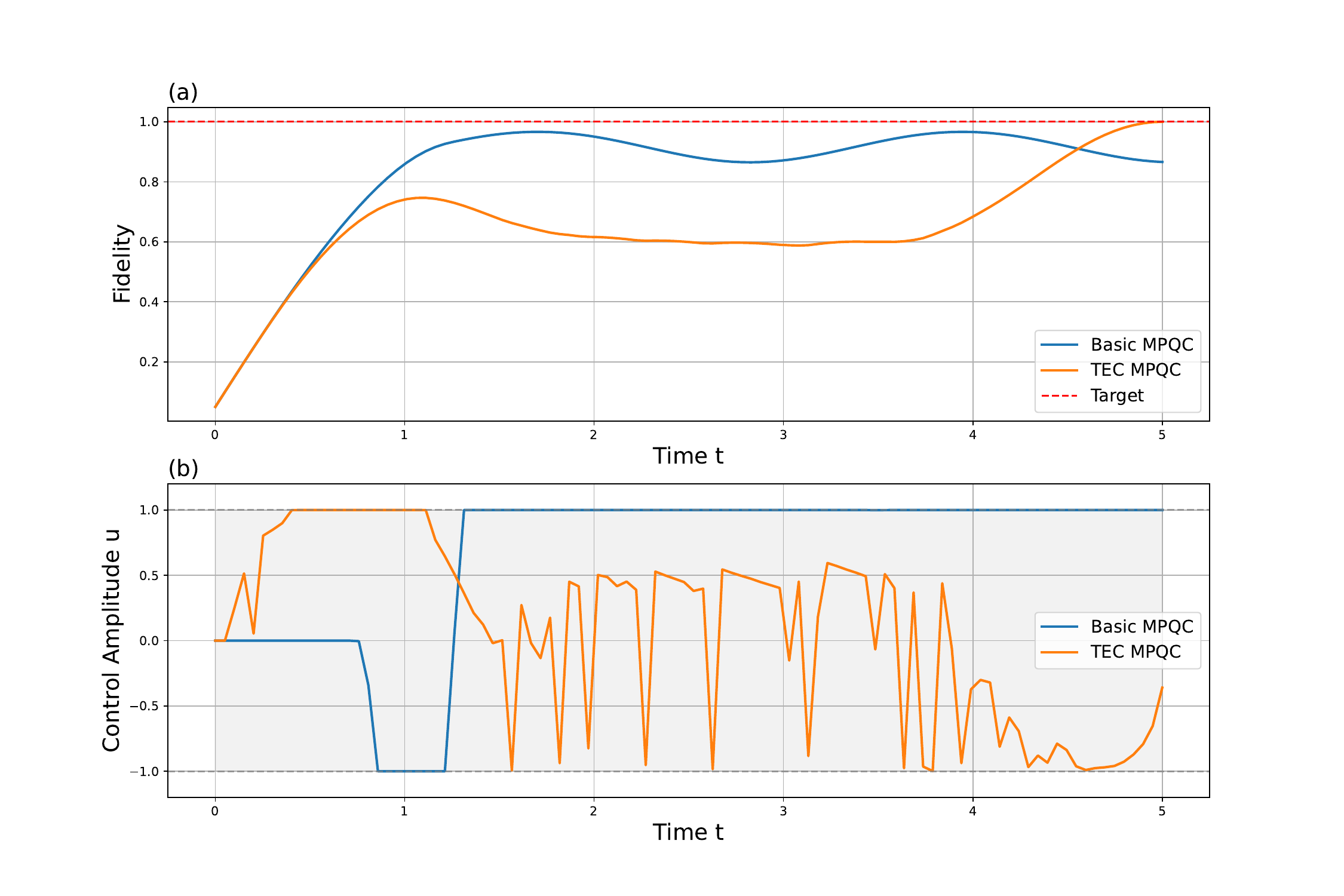} \\
    \caption{System evolution under the MPQC schemes for the single-qubit state transfer 
from $\ket{+}$ to $\ket{-}$, including (a) the fidelity and (b) the control input $u$.}
\label{fig:TEC_comparison}
\end{figure}

Figure~\ref{fig:TEC_comparison} shows that the two MPQC control inputs are of qualitatively different shapes.
The basic MPQC scheme yields a bang-bang-shaped control, 
whereas the TEC MPQC leads to
oscillatory control inputs. 
While the TEC MPQC scheme guarantees exact convergence, it does so at the expense of significantly increased computational effort.

\subsection{MPQC with setpoint optimization}
\label{subsec:tracking_mpc}
In this subsection, we show numerical results for the MPQC scheme with setpoint optimization from Section~\ref{sec:MPQC_setpoint}, which leverages an artificial steady-state 
to reduce the runtime while maintaining high fidelity.
We consider the Hamiltonian~\eqref{eq:numerical_results_TEC_Hamiltonian} with initial state $\ket{\psi_0}=\ket0$ and target state $\ket{\psi_{\mathrm{ref}}}=\ket1$.
Further, we use CasADi with IPOPT to solve the corresponding optimization problem~\eqref{eq:MPQC_tracking}.
The additional cost parameters of the MPQC problem~\eqref{eq:MPQC_tracking} are chosen as $\eta = 5$, $S = I$.

\begin{figure}[t]
    \centering 
    \includegraphics[width=0.48\textwidth]{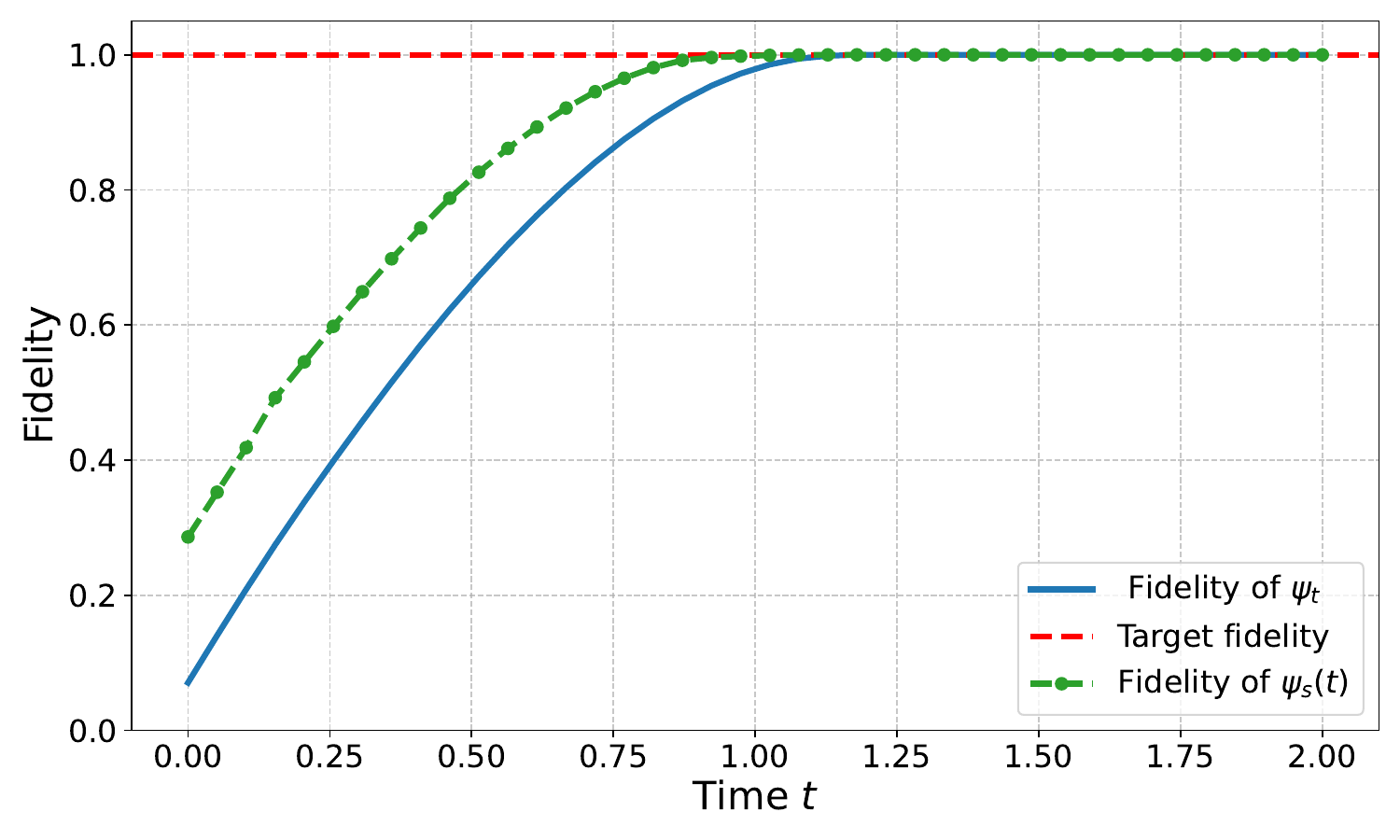}
    \caption{Single-qubit state transfer from $\ket{0}$ to $\ket{1}$ under the MPQC scheme with setpoint optimization.}    
    \label{fig:tracking_mpqc_2}
   \end{figure}

For this MPQC scheme, the minimal prediction horizon can be reduced to $L=2$ without compromising performance. Despite this drastic reduction in horizon length, the controller achieves high fidelity ($F_\mathrm{end} = 1.0000$) while keeping the total runtime at only $\Delta t_\mathrm{MPQC} = 1.23$~s. 
Figure~\ref{fig:tracking_mpqc_2} shows the fidelity of the controlled systems over time along with that of the artificial setpoint.

Using the same setup as in Figure~\ref{fig:tracking_mpqc_2}, we further investigate the benefits of applying setpoint optimization to state transfer tasks targeting other states on the Bloch sphere. The results are summarized in Table~\ref{tab:qoc_comparison3}.

\begin{table}[h]
\centering
\caption{TEC vs. Setpoint Optimization MPQC using CasADi IPOPT optimizer}
\label{tab:qoc_comparison3}
\begin{tabular}{lcccc}
\toprule
\textbf{Target} & \textbf{Scheme} & \textbf{$F_\mathrm{Final}$} & \textbf{$\Delta t_\mathrm{MPQC}[s]$} & \textbf{$L_\mathrm{min}$} \\
\midrule
$\ket{1}$ & TEC & 1.000000 & $15.33$ & $20$ \\
          & Setpoint opt. & 1.000000 & $1.23$ & $2$ \\
\midrule
$\ket{+}$ & TEC & 1.000000 & $17.42$ & $15$ \\
          & Setpoint opt. & $1.000000$ & $1.17$ & $2$ \\
\midrule
$\ket{-}$ & TEC & $1.000000$ & $16.19$ & $15$ \\
          & Setpoint opt. & $1.000000$ & $1.11$ & $2$ \\
\bottomrule
\end{tabular}
\end{table}

For each state transfer simulation, the prediction horizon was set to its minimal value that still ensures proper state convergence and satisfies the constraints on average up to a numerical tolerance of $10^{-7}$ for at least $80\%$ of the optimization steps. 
The simulations demonstrate the superior performance of the MPQC with setpoint optimization compared to the TEC MPQC approach 
when the target is an eigenstate of the controlled Hamiltonian for some control input. 
On the other hand, if this assumption is violated, then MPQC with setpoint optimization can fail to stabilize the target setpoint.

\subsection{Comparison MPQC vs.\ QOC}\label{subsec:numerical_implementation_comparison}
In this subsection, we compare the performance and runtime of MPQC to QOC.
We consider a pure state transfer problem from the ground state $\ket{0}$ to the excited state $\ket{1}$ for a two-level system with Hamiltonian
\begin{align}
    H(t) = \sigma_z +  u_1(t)\sigma_x + u_2(t)\sigma_y 
\end{align}
Figure~\ref{fig:comparison_QOC} compares the performance and computational complexity of the basic MPQC scheme from Section~\ref{sec:MPQC} vs.\ directly solving the QOC problem~\eqref{eq:QOC}.
All involved optimization problems are solved using CasADi with IPOPT.
The MPQC scheme yields a considerable speedup in comparison to QOC, which is especially pronounced for small values of the prediction horizon $L$.
Notably, this is achieved with only very little loss of performance.
In particular, Figure~\ref{fig:comparison_QOC} shows that the total cost achieved by QOC and MPQC is comparable as long as $L\geq3$.

\begin{figure}[t]
    \centering
    \includegraphics[width=0.48\textwidth]{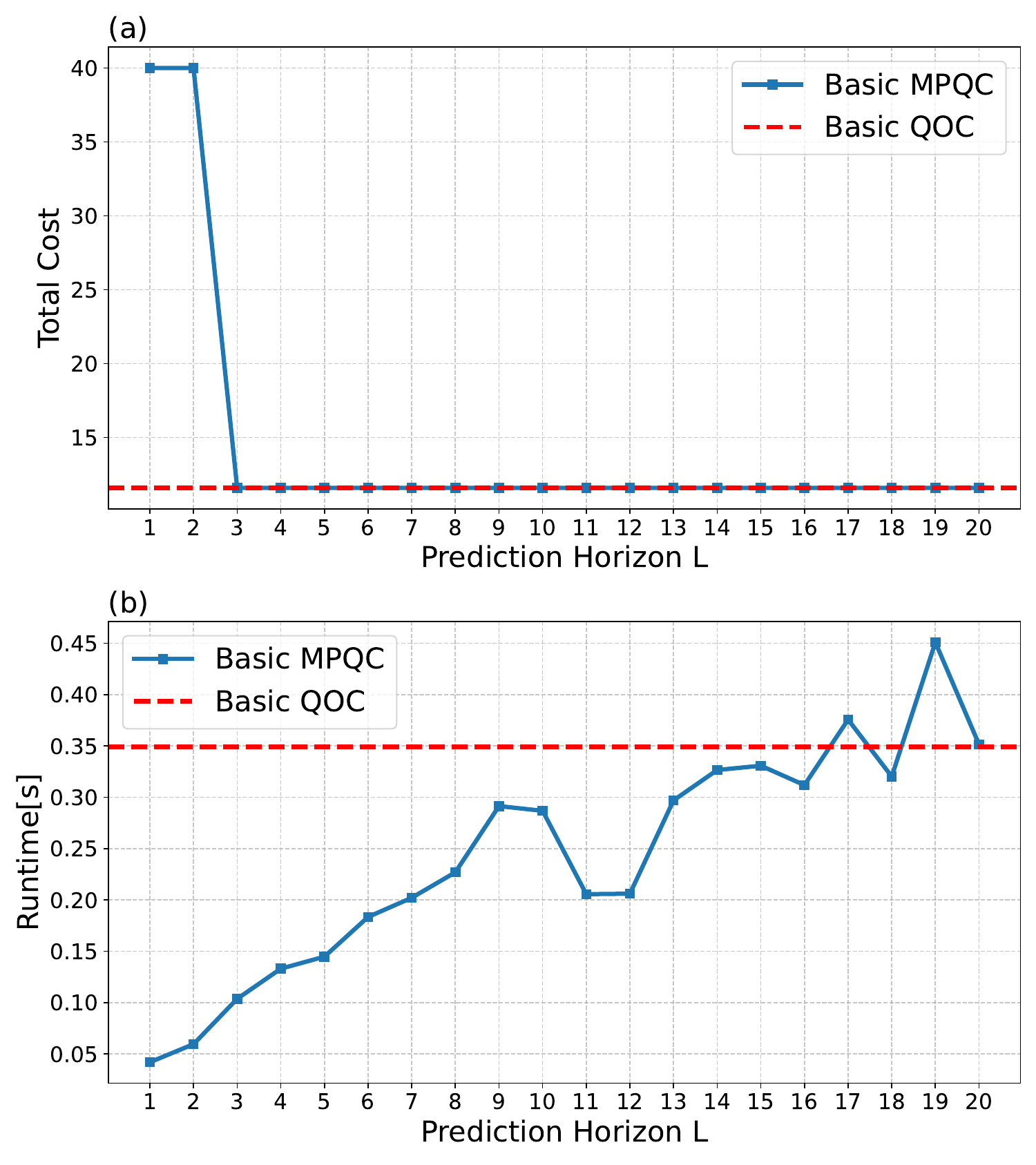}
    \caption{Comparison of basic MPQC and QOC for pure state transfer from $\ket{0}$ to $\ket{1}$.
    \textbf{(a)} Total cost achieved by basic MPQC (blue squares) as a function of the prediction horizon $L$. The red dashed line indicates the cost obtained from solving the corresponding basic QOC problem~\eqref{eq:QOC} over the full horizon $N$ directly.
\textbf{(b)} Computational runtime of basic MPQC depending on the prediction horizon and of QOC over the full horizon $N$.} 
\label{fig:comparison_QOC}
\end{figure}

\subsection{Robustness of closed-loop MPQC}
\label{subsec:robustness_of_mpqc}

In the following, we apply MPQC in a closed-loop quantum control setup, i.e., using measurements of the state $\ket{\psi_t}$ for solving the optimization problem~\eqref{eq:MPQC}.
To study the effect of model mismatch, we consider the Hamiltonian
\begin{align}\label{eq:noisy_hamiltonian}
H = (\omega + \epsilon) \sigma_x + u_1(t) \sigma_y + u_2(t) \sigma_z
\end{align}
where $\omega = 1.0$ represents the nominal drift frequency and $\epsilon \in [-1, 1]$ is an unknown error parameter affecting the $\sigma_x$ component. 
We apply both open-loop QOC based on GRAPE over $N=40$ time steps as well as the TEC MPQC scheme with horizon $L=3$, parameter $M=1$, and solved via CasADi and IPOPT.
In the respective optimization problems~\eqref{eq:QOC} and~\eqref{eq:MPQC_TEC}, the above Hamiltonian is used with $\epsilon=0$ since the precise value of $\epsilon$ is assumed to be unknown.
To compute the final fidelity obtained via the two approaches, we choose values of $\epsilon$ from a uniform grid over $[-1,1]$.

Figure~\ref{fig:robustness} shows that closed-loop TEC MPQC yields a significantly improved final fidelity in comparison to open-loop QOC when the error $\epsilon$ is non-zero.
This is due to the feedback obtained through the state measurement entering the optimization problem~\eqref{eq:MPQC_TEC} at any time.

\begin{figure}[t]
    \centering
    \includegraphics[width=0.55\textwidth]{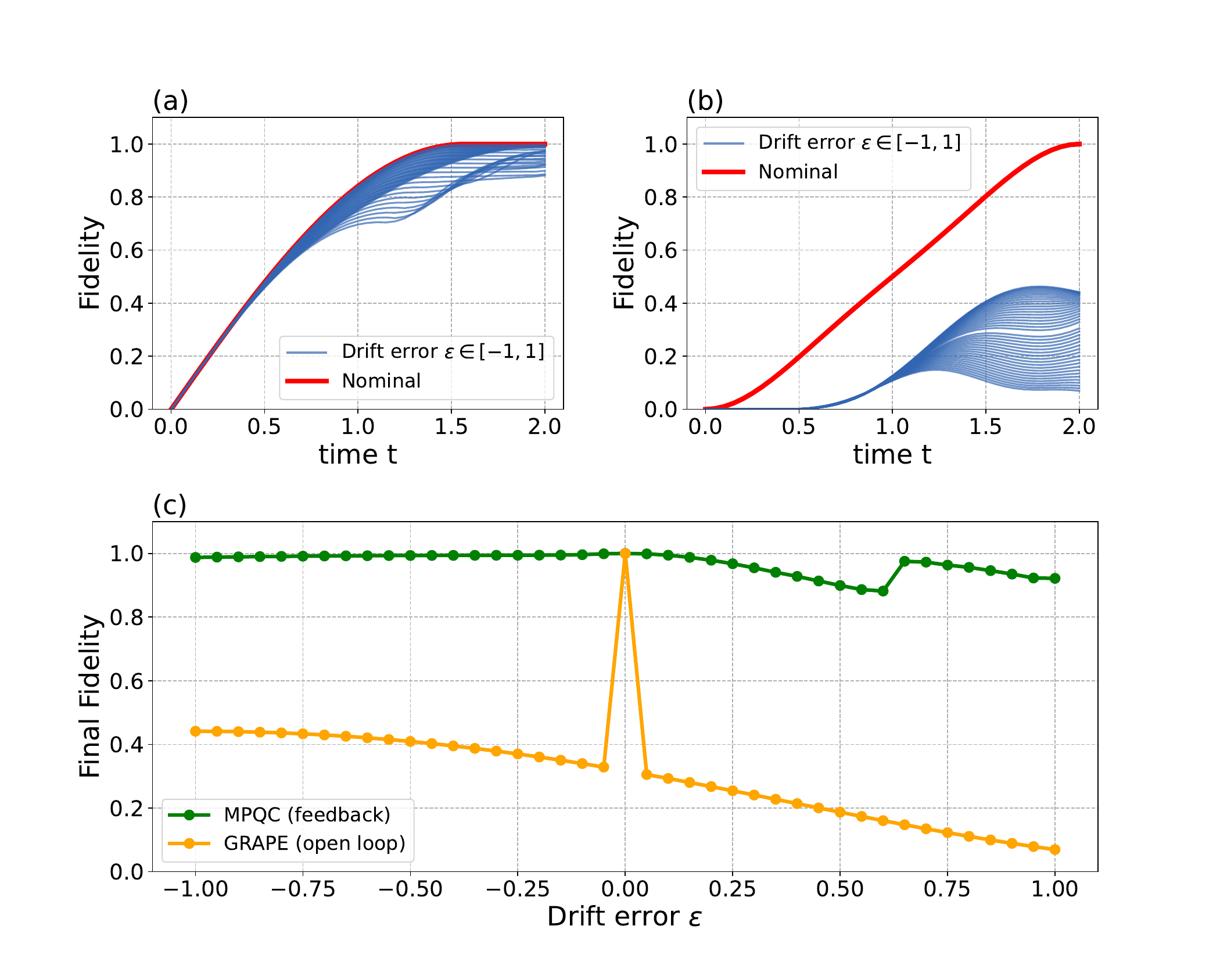}
\caption{Robustness comparison between closed-loop TEC MPQC and open-loop QOC for pure state transfer from $\ket{0}$ to $\ket{1}$ for a single-qubit system.
\textbf{(a)} Closed-loop TEC MPQC:
The red curve shows the nominal fidelity evolution in the ideal case $\epsilon=0$, whereas the blue curves depict the fidelities when the controlled quantum system evolves according to the noisy Hamiltonian~\eqref{eq:noisy_hamiltonian} for $\epsilon\in[-1,1]$.
\textbf{(b)} Open-loop GRAPE: Analogous to (a).
\textbf{(c)} Final fidelity achieved by both methods depending on the error $\epsilon \in [-1, 1]$.} 
    \label{fig:robustness}
\end{figure}

\section{Conclusion}\label{sec:discussion}

In this paper, we introduced the MPQC framework which applies MPC to QOC problems.
The key idea is to partition the original QOC problem with time horizon $N$ into multiple smaller QOC problems with horizon $L$.
This can lead to substantial improvements of efficiency and robustness in open-loop and closed-loop implementations.
We introduced an MPQC scheme based on terminal equality constraints, which admits theoretical guarantees on exponential stability, as well as MPQC with setpoint optimization, which is more practical and relies on the solid theoretical foundation of MPC for tracking.
Throughout this paper, we treated open-loop and closed-loop applications of MPQC in a unifying framework.
In particular, the presented MPQC schemes can be applied in either scenario, depending on whether the state update is obtained via simulation (open-loop) or via measurements (closed-loop).
Our numerical results demonstrate that MPQC can substantially reduce the computational complexity in comparison to standard open-loop QOC, and it can improve the robustness against model mismatch in a closed-loop implementation.

This paper provides the basis for several interesting future research directions.
In open-loop MPQC, the framework can be improved by applying more sophisticated MPC schemes to improve the practicality or allow for more general problem formulations, compare~\cite{rawlings2020model}.
Deriving theoretical guarantees for the MPQC scheme with setpoint optimization from Section~\ref{sec:MPQC_setpoint} is also relevant, especially when the target state is not an eigenstate of the controlled Hamiltonian.
Further, the presented closed-loop MPQC approach faces the limitation of relying on state tomography, leading to a possibly large overhead of required experiments.
In this regard, shadow tomography is a promising tool to reduce the number of required samples~\cite{huang2020predicting}.
Alternatively, it would be interesting to explicitly include measurement effects as in the time-optimal MPQC approaches from~\cite{lee2024robust,lee2024model,lee2025time}, or to develop alternative schemes based on weak measurements~\cite{wiseman2009quantum}.
Finally, we plan to apply the framework to larger benchmarking problems to showcase its potential advantages.

\bibliographystyle{IEEEtran}
\bibliography{references}

\appendix

\section{Details on the numerical implementation}\label{app:numerical_implementation}

The exact optimization problems introduced in Sections ~\ref{sec:MPQC}–\ref{sec:MPQC_setpoint} cannot be passed directly to standard nonlinear optimization toolboxes such as CasADi.  The main obstacle is that the discrete-time quantum dynamics are naturally expresssed in terms of complex-valued state vectors and operators, while CasADi does not natively support optimization over complex-valued functions. To address this issue, we reformulate the dynamics into an equivalent real-valued system. More precisely, a pure quantum state $\ket{\psi} \in \mathbb{C}^d$ evolving according to the Schrödinger equation
\begin{align}
    |\dot{\psi}(t)\rangle = H(u(t))\psi(t), \quad H(u(t)) = \Big( H_0 + \sum_{j=1}^m u_j(t) H_j \Big)
\end{align}
is decomposed into $\psi = a + i b$, with $a,b \in \mathbb{R}^d$. For $H(u(t)) = H_r(t) + i H_i(t)$, the dynamics translate into a coupled real-valued system for $a$ and $b$:
\begin{align}
\begin{bmatrix}
\dot{a}(t) \\[0.3em] \dot{b}(t)
\end{bmatrix}
=
\begin{bmatrix}
H_i(t) & H_r(t) \\[0.3em]
-H_r(t) & H_i(t)
\end{bmatrix}
\begin{bmatrix}
a(t) \\[0.3em] b(t)
\end{bmatrix},
\end{align}
This representation is mathematically equivalent to the original complex dynamics and can be implemented directly in CasADi.
In particular, fidelity terms $\left| \langle \psi_{\rm ref}, \psi \rangle \right|^2$ become quadratic expressions in $a$ and $b$. 
Similarly, for mixed states, the density operator $\rho \in \mathbb{C}^{d \times d}$ evolves according to the von Neumann equation
\begin{align}
    \dot{\rho}(t) = -i \big[ H(u(t)), \rho(t) \big], \qquad H(u) = H_0 + \sum_{j=1}^m u_j(t) H_j.
\end{align}
Decomposing $\rho = A + i B$, with $A,B \in \mathbb{R}^{d \times d}$, we obtain
\begin{align}
\begin{bmatrix}
\dot{A}(t) \\[0.3em] \dot{B}(t)
\end{bmatrix}
=
\begin{bmatrix}
[H_i(t), A] + [H_r(t), B] \\[0.5em]
[H_i(t), B] - [H_r(t), A]
\end{bmatrix}.
\end{align}
Vectorization of $A$ and $B$ into real vectors $a = \mathrm{vec}(A)$ and $b = \mathrm{vec}(B) \in \mathbb{R}^{d^2}$ yields a coupled real-valued system
\begin{align}
\begin{bmatrix}
\dot{a}(t) \\[0.5em] \dot{b}(t)
\end{bmatrix}
=
\mathcal{L}(u(t))
\begin{bmatrix}
a(t) \\[0.5em] b(t)
\end{bmatrix},
\label{mixed_state_dynamic}
\end{align}
where $\mathcal{L}(u)$ is the real block-structured superoperator corresponding to the commutator dynamics. 

In both cases, the reformulation step is a technical modification specific to CasADi and related optimization packages. In contrast, dedicated gradient-based quantum optimal control (QOC) algorithms such as GRAPE and Krotov’s method are designed to directly operate on the original complex-valued dynamics and therefore do not require such modifications.

\end{document}